\documentclass[10pt, conference, letterpaper]{IEEEtran}

\usepackage{amsmath}
\usepackage{amssymb}
\usepackage{float}
\usepackage{psfrag}
\usepackage[dvips]{graphicx,color}
\usepackage{cite}
\usepackage{ntheorem}

\usepackage[usenames,dvipsnames]{pstricks}
\usepackage{epsfig}
\usepackage{pst-grad} 
\usepackage{pst-plot} 

\newtheorem{theorem}{Theorem}[section]
\newtheorem{lemma}[theorem]{Lemma}
\newtheorem{proposition}[theorem]{Proposition}
\newtheorem{corollary}[theorem]{Corollary}

\newtheorem{remark}[theorem]{Remark}

\newtheorem{definition}{Definition}
\newtheorem{overview}{Overview}
\newenvironment{proof}[1][Proof]{\begin{trivlist}
\item[\hskip \labelsep {\bfseries #1}]}{\end{trivlist}}

\newcommand{\BigO}[1]{\ensuremath{\operatorname{O}\left(#1\right)}}
\begin{document}
\title{The Online Disjoint Set Cover Problem and its Applications}
\author{
\IEEEauthorblockN{Ashwin Pananjady, Vivek Kumar Bagaria} 
\IEEEauthorblockA{Department of Electrical Engineering\\
Indian Institute of Technology Madras\\
Email: \{ee10b047, ee10b025\}@ee.iitm.ac.in}
\and
\IEEEauthorblockN{Rahul Vaze}
\IEEEauthorblockA{School of Technology and Computer Science\\
Tata Institute of Fundamental Research\\
Email: vaze@tcs.tifr.res.in}
}

\maketitle

\begin{abstract}
Given a universe $U$ of $n$ elements and a collection of subsets $\mathcal{S}$ of $U$, the maximum disjoint set cover problem (DSCP) is to partition $\mathcal{S}$ into as many set covers as possible, where a set cover is defined as a collection of subsets whose union is $U$.
We consider the online DSCP, in which the subsets arrive one by one (possibly in an order chosen by an adversary), and must be irrevocably assigned to some partition on arrival with the objective of minimizing the competitive ratio. The competitive ratio of an online DSCP algorithm $A$ is defined as the maximum ratio of the number of disjoint set covers obtained by the optimal offline algorithm to the number of disjoint set covers obtained by $A$ across all inputs.
We propose an online algorithm for solving the DSCP with competitive ratio $\ln n$. We then show a lower bound of $\Omega(\sqrt{\ln n})$ on the competitive ratio for any online DSCP algorithm.
The online disjoint set cover problem has wide ranging applications in practice, including the online crowd-sourcing problem, the online coverage lifetime maximization problem in wireless sensor networks, and in online resource allocation problems.
\end{abstract}
\section{Introduction}
Consider a universe $U$ consisting of $n$ elements, i.e., $U = \{1,2,\dots, n\}$. Let $\mathcal{S}=\{S_1,S_2,\ldots\}$ be a collection of subsets of $U$, where $S_i \subseteq U$ $\forall$ $i$. Then the disjoint set cover problem (DSCP) is to find as many partitions of $\mathcal{S}$ as possible such that the union of the subsets in each partition is $U$. The DSCP is known to be NP-hard \cite{cardei2005improving}, and has an optimal approximation ratio of $\ln n$ with any polynomial time algorithm, by \cite{feige2002approximating, bagaria2013optimally}.

The DSCP is a fundamental combinatorial optimization problem that has widespread applications. Maximizing the coverage lifetime (MLCP) of a sensor network \cite{bagaria2013optimally, berman2004power} is one example. Here, $U$ is the set of targets, and each sensor $i$ can cover/track targets $S_i\subseteq U$. The objective is to find a sensor operation (on/off) schedule such that the total time for which all targets are covered is maximized. One common approach to solve the MLCP is to find the DSCP solution, and use each of the disjoint set covers in distinct time slots \cite{cardei2005improving, bagaria2013optimally, slijepcevic2001power}.

Another DSCP application of interest is in resource allocation and scheduling problems \cite{pananjady2014maximizing}. A canonical example of a resource allocation problem is where the universe $U$ represents the set of files or sub-files of a movie or large file, and each server $i$ contains a subset $S_i$ of those files. Each server has limited capability, and can serve at most one user at any given time. Thus, maximizing the number of users that can access all files of the movie (and consequently the revenue) is equivalent to solving the DSCP.

The new paradigm of crowd-sourcing also involves solving the DSCP. In crowd-sourcing \cite{vazeashwin, sheng2012energy}, a platform (public utility) advertises the set of tasks (the universe $U$) that it wants to be accomplished (e.g. pothole tracking, community service, etc.). Each user $i$ submits a list of tasks $S_i$ that it can perform, and the platform has to group/cluster subset of users in as many groups as possible such that each group ensures task coverage, i.e. all tasks should be performed by at least one of the users in the group. The crowd-sourcing problem with task coverage is equivalent to the DSCP. 

The DSCP is also relevant for finding efficient supply chain management solutions \cite{lu2011fundamentals}, where $n$ distinct raw materials/sub-tasks are required for producing a particular good by a machine, and each sub-contractor/auxiliary-machine $i$ can supply/accomplish only a subset of raw materials/sub-tasks $S_i$ \cite{tompkins2003optimization}. Finding the optimal allocation/routing of supplier/sub-tasks to distinct machines in order to maximize the number of simultaneously producible goods is equivalent to the DSCP. 

In this paper, we consider an online version of the DSCP, where subsets arrive sequentially in time, and each subset has to be assigned to a partition irrevocably without knowledge of future subset arrivals. The objective is similar to the earlier case: to maximize the number of partitions such that the union of subsets in each partition is equal to $U$ at the end of all subset arrivals, but now relative to the optimal offline algorithm. The offline algorithm refers to the case when the sequence of arrival of subsets is revealed to the algorithm in advance. To study the effectiveness of online algorithms, the competitive ratio, which measures the ratio of the profit of the optimal offline algorithm and a particular online algorithm, is the metric of choice \cite{borodin1998online}. The smaller the competitive ratio of an online algorithm, the better its performance. Note that here we make no assumption on the complexity of the offline/online algorithm, the optimal offline or any online algorithm is allowed to have exponential complexity. 

Studying the online version of DSCP is important, since finding an efficient online algorithm for the DSCP is equivalent to solving the online versions of the optimization problems defined above. For example, the online DSCP corresponds to the online MLCP in a natural way, where sensors arrive/wake-up sequentially in time, and each sensor's on-off schedule has to be decided in an online fashion so as to maximize the coverage lifetime. Similarly, the online DSCP corresponds to online crowd-sourcing, where users submit their requests sequentially in time, and must be grouped with other users irrevocably without knowledge of task sets of future users. 

There is also a natural correspondence between the online DSCP and the online resource allocation problem, where servers become active or acquire the subset of files at arbitrary times, and server-user association has to be done in an online fashion without knowledge of future server arrivals. The online version of supply chain management has a similar correspondence to the online DSCP. 

Let $F_{min}$ be the minimum number of times any universe element occurs across all the subsets. In this paper, we make the following contributions.
\begin{itemize}
\item We show that any online algorithm must be given $F_{min}$ in advance to perform with a non-trivial competitive ratio.
\item Through the online polychromatic coloring of a specially defined hypergraph, we propose an online algorithm for the online DSCP that has a competitive ratio of $\BigO{\ln n}$.
\item We also show a lower bound on the competitive ratio, of $\Omega(\sqrt{\ln n})$, i.e. no online algorithm can have a better competitive ratio than $\Omega(\sqrt{\ln n})$. Note that no complexity assumptions are made on the online algorithm, and the lower bound holds for online algorithms even with exponential complexity. 
\end{itemize}

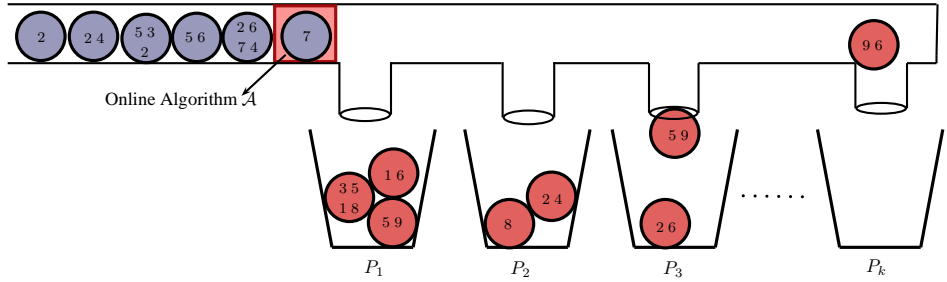
\begin{figure*}[ht]
\centering

\scalebox{0.7} 
{
\begin{pspicture}(0,-2.6489062)(17.79,2.6289062)

\definecolor{color4062b}{rgb}{0.8823529411764706,0.3803921568627451,0.3686274509803922}
\definecolor{color4108b}{rgb}{0.9921568627450981,0.5843137254901961,0.5843137254901961}
\definecolor{color4108}{rgb}{0.6431372549019608,0.050980392156862744,0.050980392156862744}
\definecolor{color53b}{rgb}{0.6,0.6,0.7333333333333333}
\rput{-90.0}(14.6110935,18.288906){\pscircle[linewidth=0.06,dimen=outer,fillstyle=solid,fillcolor=color4062b](16.45,1.8389063){0.49}}
\psline[linewidth=0.06cm](5.74,0.22890624)(6.16,-2.0110939)
\psline[linewidth=0.06cm](6.16,-2.0110939)(7.7,-2.0110939)
\psline[linewidth=0.06cm](7.7,-2.0110939)(8.12,0.22890624)
\psline[linewidth=0.06cm](8.68,0.22890624)(9.1,-2.0110939)
\psline[linewidth=0.06cm](9.1,-2.0110939)(10.64,-2.0110939)
\psline[linewidth=0.06cm](10.64,-2.0110939)(11.06,0.22890624)
\psline[linewidth=0.06cm](11.48,0.22890624)(11.9,-2.0110939)
\psline[linewidth=0.06cm](11.9,-2.0110939)(13.44,-2.0110939)
\psline[linewidth=0.06cm](13.44,-2.0110939)(13.86,0.22890624)
\psline[linewidth=0.06cm](15.38,0.22890624)(15.8,-2.0110939)
\psline[linewidth=0.06cm](15.8,-2.0110939)(17.34,-2.0110939)
\psline[linewidth=0.06cm](17.34,-2.0110939)(17.76,0.22890624)
\psline[linewidth=0.04cm](0.0,2.6089063)(17.64,2.6089063)
\psline[linewidth=0.04cm](17.66,2.6089063)(17.64,1.6289062)
\psline[linewidth=0.04cm](6.3,1.5085616)(6.3,0.51136315)
\psline[linewidth=0.04cm](7.28,1.5085616)(7.28,0.51136315)
\psellipse[linewidth=0.04,dimen=outer](6.79,0.51136315)(0.49,0.14245692)
\psline[linewidth=0.04cm](9.4,1.508555)(9.4,0.45886236)
\psline[linewidth=0.04cm](10.38,1.508555)(10.38,0.45886236)
\psellipse[linewidth=0.04,dimen=outer](9.89,0.45886236)(0.49,0.1499561)
\psline[linewidth=0.04cm](16.1,1.4889143)(16.1,0.63140726)
\psline[linewidth=0.04cm](17.08,1.4889143)(17.08,0.63140726)
\psellipse[linewidth=0.04,dimen=outer](16.59,0.63140726)(0.49,0.122501)
\psframe[linewidth=0.06,linecolor=color4108,dimen=outer,fillstyle=solid,fillcolor=color4108b](6.26,2.6089063)(5.04,1.4889063)
\pscircle[linewidth=0.06,dimen=outer,fillstyle=solid,fillcolor=color53b](1.63,1.9989063){0.49}

\usefont{T1}{ptm}{m}{n}
\rput(1.6435938,1.9689063){\footnotesize $2\;4$}
\pscircle[linewidth=0.06,dimen=outer,fillstyle=solid,fillcolor=color53b](2.63,1.9589063){0.49}

\usefont{T1}{ptm}{m}{n}
\rput(2.6035938,1.6889062){\footnotesize $2$}

\usefont{T1}{ptm}{m}{n}
\rput(2.6035938,2.0689063){\footnotesize $5\;3$}
\pscircle[linewidth=0.06,dimen=outer,fillstyle=solid,fillcolor=color53b](3.59,1.9789063){0.49}

\usefont{T1}{ptm}{m}{n}
\rput(3.5635939,1.9889063){\footnotesize $5\;6$}
\pscircle[linewidth=0.06,dimen=outer,fillstyle=solid,fillcolor=color53b](4.55,1.9789063){0.49}
\pscircle[linewidth=0.06,dimen=outer,fillstyle=solid,fillcolor=color53b](0.63,1.9989063){0.49}

\usefont{T1}{ptm}{m}{n}
\rput(0.6435937,1.9889063){\footnotesize $2$}
\pscircle[linewidth=0.06,dimen=outer,fillstyle=solid,fillcolor=color4062b](7.31,-1.5410937){0.49}
\pscircle[linewidth=0.06,dimen=outer,fillstyle=solid,fillcolor=color4062b](6.49,-1.0610938){0.49}
\pscircle[linewidth=0.06,dimen=outer,fillstyle=solid,fillcolor=color4062b](7.33,-0.60109377){0.49}
\pscircle[linewidth=0.06,dimen=outer,fillstyle=solid,fillcolor=color4062b](9.53,-1.5610938){0.49}
\pscircle[linewidth=0.06,dimen=outer,fillstyle=solid,fillcolor=color4062b](10.33,-1.0410937){0.49}
\pscircle[linewidth=0.06,dimen=outer,fillstyle=solid,fillcolor=color4062b](12.49,-1.5610938){0.49}
\rput{-90.0}(12.511094,12.828906){\pscircle[linewidth=0.06,dimen=outer,fillstyle=solid,fillcolor=color4062b](12.67,0.15890625){0.49}}
\pscircle[linewidth=0.06,dimen=outer,fillstyle=solid,fillcolor=color53b](5.67,1.9989063){0.49}
\psline[linewidth=0.06cm,linestyle=dotted,dotsep=0.16cm](13.94,-1.0310937)(15.12,-1.0510937)

\usefont{T1}{ptm}{m}{n}
\rput(4.543594,2.1689062){\footnotesize $2\;6$}

\usefont{T1}{ptm}{m}{n}
\rput(5.6835938,2.0089064){\footnotesize $7$}

\usefont{T1}{ptm}{m}{n}
\rput(6.483594,-0.89109373){\footnotesize $3\;5$}

\usefont{T1}{ptm}{m}{n}
\rput(6.483594,-1.2710937){\footnotesize $1\;8$}

\usefont{T1}{ptm}{m}{n}
\rput(7.3435936,-0.6510937){\footnotesize $1\;6$}

\usefont{T1}{ptm}{m}{n}
\rput(7.3035936,-1.5310937){\footnotesize $5\;9$}

\usefont{T1}{ptm}{m}{n}
\rput(12.743594,0.10890625){\footnotesize $5\;9$}

\usefont{T1}{ptm}{m}{n}
\rput(10.343594,-1.0910938){\footnotesize $2\;4$}

\usefont{T1}{ptm}{m}{n}
\rput(12.503593,-1.6310937){\footnotesize $2\;6$}

\usefont{T1}{ptm}{m}{n}
\rput(9.503593,-1.5710938){\footnotesize $8$}
\psline[linewidth=0.04cm](12.179969,1.4892693)(12.179969,0.54362833)
\psline[linewidth=0.04cm](13.1208,1.4892693)(13.1208,0.54362833)
\psellipse[linewidth=0.04,dimen=outer](12.650384,0.54362833)(0.47041565,0.13509157)

\usefont{T1}{ptm}{m}{n}
\rput(4.563594,1.7889062){\footnotesize $7\;4$}

\usefont{T1}{ptm}{m}{n}
\rput(6.961406,-2.4060938){$P_1$}

\usefont{T1}{ptm}{m}{n}
\rput(9.741406,-2.4260938){$P_2$}

\usefont{T1}{ptm}{m}{n}
\rput(12.641406,-2.4260938){$P_3$}

\usefont{T1}{ptm}{m}{n}
\rput(16.501406,-2.3860939){$P_k$}
\psline[linewidth=0.04cm,arrowsize=0.05291667cm 2.0,arrowlength=1.4,arrowinset=0.4]{->}(5.28,1.6489062)(4.44,0.86890626)

\usefont{T1}{ptm}{m}{n}
\rput(3.2895312,0.77390623){Online Algorithm $\mathcal{A}$}
\psline[linewidth=0.04cm](0.0,1.4889063)(6.3,1.4889063)
\psline[linewidth=0.04cm](7.28,1.4889063)(9.38,1.4889063)
\psline[linewidth=0.04cm](10.36,1.4889063)(12.18,1.4889063)
\psline[linewidth=0.04cm](13.12,1.4889063)(16.1,1.4889063)
\psline[linewidth=0.04cm](17.08,1.4889063)(17.64,1.4889063)
\psline[linewidth=0.04cm](17.64,1.4889063)(17.64,1.6289062)

\usefont{T1}{ptm}{m}{n}
\rput(16.423594,1.8289063){\footnotesize $9\;6$}
\end{pspicture}
}
\caption{An online allocation $\mathcal{M}_\mathcal{A}$ by an online algorithm $\mathcal{A}$ for the universe $U=\{1,2\ldots, 9\}$ and a sequence of subsets $\mathbf{S}$ depicted by balls. Notice that $P_1$ requires the elements $2$, $4$ and $7$ to form a set cover. Also, red balls can no longer be reallocated. $T(\mathcal{M}^*_{off}, \mathbf{S})=2$ and $T(\mathcal{M}_{\mathcal{A}}, \mathbf{S})\leq 1$, since $F_1=2$. Note also that the online algorithm may choose to either place a ball along with others in a pre-existing bin, or create another bin.} \label{ballsbins}
\end{figure*}

\section{Model and Preliminaries}
The universe of \emph{elements} is represented by $U=\{1,2,3,\ldots,n-1,n\}$, unless stated otherwise. We denote the collection of subsets provided by $\mathcal{S}=\{S_1,S_2,\ldots\}$, where $S_j\subseteq U$ $\forall$ $j$. The number of subsets is therefore $|\mathcal{S}|$. We define a set cover as a collection of subsets $C\subseteq\mathcal{S}$ such that $\cup_{S_i\in C}S_i=U$. We assume that at least one set cover exists in $\mathcal{S}$, i.e., $\cup_{S_i\in\mathcal{S}}S_i=U$.

Note that we are interested in an \emph{online} scenario, in which subsets arrive one at a time, and so the order of the subsets becomes important. In order to take this into account, we define an ordered tuple of subsets, which we will call a subset sequence, by $\mathbf{S}=[S_1,S_2,\ldots]$. Note that a permutation of $\mathbf{S}$ is distinct from $\mathbf{S}$. When $\mathbf{S}$ is mentioned in the context of an offline algorithm however, it is equivalent to $\mathcal{S}$, since all subset arrivals are known in advance to an offline algorithm. We also define the \emph{concatenation} of two subset sequences $\mathbf{S}_a=[S^a_1,S^a_2,\ldots, S^a_x]$ and $\mathbf{S}_b=[S^b_1,S^b_2,\ldots, S^b_y]$ by $\mathbf{S}_a\wedge \mathbf{S}_b = [S^a_1,S^a_2,\ldots, S^a_x, S^b_1,S^b_2,\ldots, S^b_y]$. Similarly, $\mathbf{S}\wedge S$ for some subset $S$ corresponds to adding $S$ to the end of subset sequence $\mathbf{S}$. For the creation of $\mathbf{S}$ from subsets, we use the familiar ordered set notation, e.g. $\mathbf{S}=[S_j:S_j = \{j\}, j\in U]$ is effectively $\mathbf{S}=[\{1\},\{2\},\ldots, \{n\}]$.

Define an \emph{allocation} $\mathcal{M}$ as a partition of $\mathcal{S}$ into $P_1,P_2,\ldots$. An allocation made by an algorithm $\mathcal{A}$ on a sequence of subsets $\mathbf{S}$ is denoted by $\mathcal{M}_{\mathcal{A}}(\mathbf{S})$. We would like to make an allocation $\mathcal{M}$ such that the maximum number of $P_i$s form set covers. This is equivalent to solving the \emph{maximum disjoint set cover problem} (DSCP) \cite{bagaria2013optimally}, defined as follows:

\begin{definition}[DSCP \cite{bagaria2013optimally}] \label{def_DSCP}
Given a universe $U$ and a set of subsets $\mathcal{S}$, find as many set covers $C$ as possible such that they are all pairwise disjoint (i.e. $C_i\cap C_j = \phi$ $\forall$ $i\neq j$).
\end{definition}

In the \emph{offline} case of the DSCP, an allocation is made with knowledge of $\mathcal{S}$ in its entirety.
However, in this paper, we are interested in the scenario in which the algorithm is forced to make an allocation \emph{online}. An online algorithm must assign a subset to one of the partitions $P_i$ as soon as it arrives, and this assignment cannot be changed subsequently. The objective is to maximize the number of $P_i$s that form set covers at the end of all subset arrivals, relative to the optimal offline algorithm.

As an analogy, one can picture the online version of the DSCP as in Figure \ref{ballsbins}. Subsets are represented by balls containing elements from $U$, and partitions by bins. The online algorithm has the objective of assigning balls to bins, or in other words, making an allocation. The balls arrive one by one, and the algorithm must drop each into a particular bin as and when it arrives. The objective is to ensure that after all arrivals, the number of bins that form set covers is maximized, relative to the optimal offline allocation.

Define $T(\mathcal{M},\mathbf{S})$ as the number of $P_i$s that form set covers after an online allocation $\mathcal{M}$ is made on a subset sequence $\mathbf{S}$. We also define $\mathcal{M}^*_{off}(\mathbf{S})$ as the offline allocation that maximizes $T(\mathcal{M},\mathbf{S})$, and call this maximum $T(\mathcal{M}^*_{off},\mathbf{S})$.

Since, $\cup_{S_i\in\mathbf{S}}S_i=U$, the following is trivially established:
\begin{equation}
T(\mathcal{M}^*_{off},\mathbf{S})\geq 1.
\end{equation}

To analyse the performance of online algorithms for the DSCP, a figure of merit is the \emph{competitive ratio} \cite{borodin1998online}. In this paper, we will denote the competitive ratio of an online algorithm $\mathcal{A}$ on a sequence of subsets $\mathbf{S}$ by $\mu_{\mathbf{S}}(\mathcal{A})$, defined as
\begin{equation}
\mu_{\mathbf{S}}(\mathcal{A})=\frac{T(\mathcal{M}^*_{off},\mathbf{S})}{T(\mathcal{M}_{\mathcal{A}},\mathbf{S})}.
\end{equation}
Note that $\mu_{\mathbf{S}}(\mathcal{A})\geq 1$ $\forall$ $\mathbf{S}, \mathcal{A}$. Also define the worst case competitive ratio $\mu(\mathcal{A})$ of an algorithm $\mathcal{A}$ as its highest competitive ratio over all subset sequences.
\begin{equation}
\mu(\mathcal{A})=\max_{\mathbf{S}}\mu_{\mathbf{S}}(\mathcal{A})=\max_{\mathbf{S}}\frac{T(\mathcal{M}^*_{off},\mathbf{S})}{T(\mathcal{M}_{\mathcal{A}},\mathbf{S})}. \label{compratio}
\end{equation} 

Our objective is to design online algorithms with minimum worst-case competitive ratio. From hereon, the term competitive ratio will be used to mean worst case competitive ratio, unless mentioned otherwise. A good online algorithm $\mathcal{A}$ will have $\mu(\mathcal{A})$ close to unity; an online algorithm $\mathcal{A}$ with $\mu(\mathcal{A})$ increasing with the input parameters is not desirable.

We now introduce some notation that will be used in the rest of this paper. Given $U$ and $\mathcal{S}$ (or $\mathbf{S}$), it is useful to define the frequency $F_i(\mathcal{S})$ of element $i\in U$ as the number of subsets in $\mathcal{S}$ that it appears in, i.e. $F_i(\mathcal{S})=\#\{S_j\in \mathcal{S}:i\in S_j\}$. We also define $F_{min}(\mathcal{S})=\min_iF_i(\mathcal{S})$. In order to simplify notation, we will use just $F_i$ and $F_{min}$ when there is no ambiguity about the sequence of subsets under consideration.

It is easy to see that $F_{min}$ is an upper bound on the optimal solution of the DSCP, since each set cover must contain the element with frequency $F_{min}$ and all set covers must be disjoint. Therefore, for all algorithms $\mathcal{A}$:
\begin{equation}
T(\mathcal{M}_{\mathcal{A}},\mathbf{S})\leq T(\mathcal{M}^*_{off},\mathbf{S})\leq F_{min}. \label{upperbound}
\end{equation}

Finding $T(\mathcal{M}^*_{off},\mathbf{S})$ optimally for arbitrary $\mathbf{S}$ (or providing an offline algorithm to the DSCP) is NP-complete \cite{cardei2005improving}. In this paper, however, we will calculate the competitive ratio of our online algorithms $\mu$ without restricting the optimal offline algorithm to run in polynomial time. We will therefore take $T(\mathcal{M}^*_{off},\mathbf{S})$ to be the number of disjoint set covers returned by the optimal offline algorithm, notwithstanding its time complexity. However, finding even an expression for $T(\mathcal{M}^*_{off},\mathbf{S})$ combinatorially is not tractable, and so finding $\mu$ exactly is difficult. We will therefore use \eqref{upperbound} to bound $\mu$.

We will now analyse in the next section exactly how much information an online DSCP algorithm would need in advance in order for it to perform with a non-trivial competitive ratio.
\section{Essential Information for an Online DSCP Algorithm}
Before going into the crux of this section, we first introduce the trivial online algorithm \texttt{GreedyCover} $\mathcal{V}$.
\begin{definition}[\texttt{GreedyCover} $\mathcal{V}$]
The allocation $\mathcal{M}_{\mathcal{V}}$ is made such that it always completes a set cover before moving on to the next one. In other words, each incoming subset is placed in partition $P_i$ until $P_i$ becomes a set cover, after which the next subset is placed in $P_{i+1}$. \label{trivalgo}
\end{definition}  

\begin{lemma}
The competitive ratio of $\mathcal{V}$, $\mu(\mathcal{V})\leq F_{min}$. \label{trivial}
\end{lemma}
\begin{proof}
Algorithm $\mathcal{V}$ will always return at least one set cover, since $\cup_{S_i\in\mathbf{S}}S_i=U$ $\forall$ $\mathbf{S}$. In other words, $T(\mathcal{M}_{\mathcal{V}},\mathbf{S}) \geq 1$ $\forall$ $\mathbf{S}$. The proof of Lemma \ref{trivial} follows from \eqref{upperbound}, since $F_{min}$ is an upper bound on $T(\mathcal{M}^*_{off},\mathbf{S})$.
\end{proof}

Similarly, it is clear that $\mu(\mathcal{A})\leq F_{min}$ for all algorithms $\mathcal{A}$ that always produce at least one set cover. We will now show that no online algorithm can perform better than algorithm $\mathcal{V}$ with prior knowledge of just $U$ and $|\mathcal{S}|$.
\begin{theorem}
The competitive ratio of any online algorithm  is lower bounded by $F_{min}$, i.e., $\mu(\mathcal{A})\geq F_{min}$ even when $\mathcal{A}$ is given the universe $U$ and number of subsets $|\mathcal{S}|$ a-priori.  \label{minreqd}
\end{theorem}
\begin{proof}
We will present two sequences of subsets $\mathbf{S}_1$ and $\mathbf{S}_2$ such that $\min_{\mathcal{A}}\max\Big(\mu_{\mathbf{S}_1}(\mathcal{A}), \mu_{\mathbf{S}_2}(\mathcal{A})\Big)= F_{min}$, where the minimum is taken over all online algorithms $\mathcal{A}$. This would serve as a proof of Theorem \ref{minreqd}, by \eqref{compratio}.

Let the universe be $U= \{1,2,\ldots, n\}$. This is given to the online algorithm a-priori, along with $|\mathcal{S}|$. Therefore, $|\mathbf{S}_1|=|\mathbf{S}_2|=|\mathcal{S}|$ is fixed. Let $|\mathcal{S}|$ be sufficiently large.

Define $\mathbf{S}_{com}= [\{1,2\}, \{1,3\}, \{1,4\},\ldots, \{1,n\}]$. Note that $|\mathbf{S}_{com}|=n-1$. Also define $\mathbf{S}_{1r}= [\{1\},\{1\},\ldots |\mathcal{S}|-(n-1) \text{ times}]$. Now, let $\mathbf{S}_1= \mathbf{S}_{com}\wedge \mathbf{S}_{1r}$. Note that $T(\mathcal{M}^*_{off}, \mathbf{S}_1) = 1$, since a set cover can be obtained by the combination of all the subsets in  $\mathbf{S}_{com}$.

Now, let $\mathbf{S}_{2r}=[S_j:S_j = U\setminus \{1,j+1\}, 1\leq j \leq n-1]$. Also define $\mathbf{S}_{2e}= [\{2\},\{2\},\ldots |\mathcal{S}|-(2n-2) \text{ times}]$. Let $\mathbf{S}_2= \mathbf{S}_{com}\wedge \mathbf{S}_{2r}\wedge \mathbf{S}_{2e}$. It is clear that $T(\mathcal{M}^*_{off}, \mathbf{S}_2) = n-1 = F_{min}(\mathbf{S}_2)$, since the $j$th subset of $\mathbf{S}_{2r}$ is the complement of the $j$th subset in $\mathbf{S}_{com}$ with respect to $U$.

Notice that the first $n-1$ subsets of both $\mathbf{S}_1$ and $\mathbf{S}_2$ are identical, and represented by the sequence $\mathbf{S}_{com}$. Consider two classes of online algorithms represented by $\mathcal{X}$ and $\mathcal{Y}$. Let any algorithm $\mathcal{A}\in \mathcal{X}$ make an allocation $\mathcal{M}_{\mathcal{A}}$ such that all subsets in $\mathbf{S}_{com}$ end up in the same partition, and let any algorithm $\mathcal{B}\in \mathcal{Y}$ make an allocation $\mathcal{M}_{\mathcal{A}}$ such that all subsets in $\mathbf{S}_{com}$ do not end up in the same partition. Note that the classes $\mathcal{X}$ and $\mathcal{Y}$ are disjoint and together span all online algorithms for the DSCP. The following relations become immediately clear:

\begin{align}
\mu_{\mathbf{S}_1}(\mathcal{A}) &= 1 \text{  }\forall \text{ }\mathcal{A} \in \mathcal{X}, \nonumber \\
\mu_{\mathbf{S}_1}(\mathcal{B}) &= \infty \text{  }\forall \text{ }\mathcal{B} \in \mathcal{Y}, \nonumber\\
\mu_{\mathbf{S}_2}(\mathcal{A}) &= F_{min} \text{  }\forall \text{ }\mathcal{A} \in \mathcal{X}, \nonumber\\
\mu_{\mathbf{S}_2}(\mathcal{B}) &= 1 \text{ for some }\mathcal{B} \in \mathcal{Y}. \nonumber
\end{align}
We therefore arrive at the fact that 
\begin{align}
\min_{\mathcal{A}\in \mathcal{X}} \max\Big(\mu_{\mathbf{S}_1}(\mathcal{A}), \mu_{\mathbf{S}_2}(\mathcal{A})\Big) &= F_{min}, \\
\min_{\mathcal{B}\in \mathcal{Y}} \max\Big(\mu_{\mathbf{S}_1}(\mathcal{B}), \mu_{\mathbf{S}_2}(\mathcal{B})\Big) &= \infty.
\end{align}
This shows that
\begin{equation}
\min_{\mathcal{C}} \max\Big(\mu_{\mathbf{S}_1}(\mathcal{C}), \mu_{\mathbf{C}_2}(\mathcal{C})\Big) = F_{min},
\end{equation}
where the minimum over $\mathcal{C}=\mathcal{X}\cup \mathcal{Y}$ is taken over all online algorithms for the DSCP. The proof is therefore complete. We also point out here that $F_{min}(\mathbf{S}_1)=1$ and $F_{min}(\mathbf{S}_2)=n-1$, so if the online algorithm was provided with $F_{min}$ it could have chosen to use just one partition for $\mathbf{S}_1$ and $n-1$ partitions for $\mathbf{S}_2$, thus improving its worst case competitive ratio.
\end{proof}

No online algorithm can therefore outperform the trivial \texttt{GreedyCover} algorithm $\mathcal{V}$ (Definition \ref{trivalgo}) with knowledge of just $U$ and $|\mathcal{S}|$, from Lemmas \ref{trivial} and \ref{minreqd}. One can similarly show that even if $\max_iF_i$ or $\frac{1}{n}\sum_iF_i$ are provided along with $U$ to an online algorithm, its competitive ratio is lower bounded by $F_{min}$. We will therefore assume the following.
\begin{remark}
From hereon, all online algorithms know $U$ and $F_{min}$ a-priori. \label{mininf}
\end{remark}
In the next section, we we will construct an online algorithm that is provided $F_{min}$ in advance and performs with a non-trivial competitive ratio.
\section{An Online DSCP Algorithm}
We will use the offline DSCP algorithm of \cite{bagaria2013optimally}, and make modifications such that it becomes an online algorithm. We will first present a few ideas and the algorithm from \cite{bagaria2013optimally}.

\subsection{Hypergraph representation}
The universe $U$ and set of subsets $\mathcal{S}$ can also be represented as a hypergraph, as shown in \cite{bollobas2013cover}. Define a hypergraph $\mathcal{H}_{U,\mathcal{S}}(V,E)$ representing a universe $U$ and subsets $\mathcal{S}$, having vertex set $V$ and hyperedge set $E$ as follows:
\begin{definition}[$\mathcal{H}_{U,\mathcal{S}}(V,E)$]
Each subset $S_j\in \mathcal{S}$ is represented by a vertex $v_j\in V$. A hyperedge $e_i\in E$ contains vertex $v_j$ if element $i\in U$ is such that $i\in S_j$.
\end{definition}
While dealing with a sequence of subsets $\mathbf{S}$, we will denote the corresponding hypergraph representation by $\mathcal{H}_{U,\mathbf{S}}$.

Figure \ref{Dual} illustrates an example construction of $\mathcal{H}_{U,\mathcal{S}}(V,E)$ from $U$ and $\mathcal{S}$. Note that there are $|\mathcal{S}|$ vertices and $n$ hyperedges in $\mathcal{H}_{U,\mathcal{S}}(V,E)$. The set of vertices in hyperedge $e_i$ is denoted by $V(e_i)$, and so $|V(e_i)|=F_{i}$. Hence, the smallest number of vertices in any hyperedge $\min_i|V(e_i)|=F_{min}$.
\begin{figure}[H]
\centering
	\hspace{10mm}
    \psfrag{1}[][l][0.9]{$1$}
    \psfrag{2}[][r][0.9]{$2$}
    \psfrag{3}[][][0.9]{$3$}
    \psfrag{4}[][b][0.9]{$4$}
    \psfrag{b}[][b][0.9]{ }
    \psfrag{\{1,2,4\}}[][][.55]{$\{1,2,4\}$}
    \psfrag{\{1,3\}}[][][.6]{$\{1,3\}$}
    \psfrag{\{2,3\}}[][][.6]{$\{2,3\}$}
    \includegraphics[width=.25\textwidth]{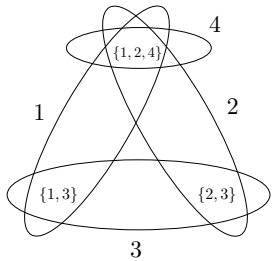} \\
  \caption{Hypergraph $\mathcal{H}_{U,\mathcal{S}}$ for $\mathcal{U} = \{ 1,2,3,4\} $ and $ \mathcal{S} = \{ \{ 1,2,4\}, \{2,3\}, \{1,3\} \}$} \label{Dual}
\end{figure} 
The DSCP on $U$ and $\mathcal{S}$ can be solved offline on $\mathcal{H}_{U,\mathcal{S}}$ by an operation known as \emph{polychromatic colouring}.
\begin{definition}[Polychromatic Colouring]
Colour the vertices of the hypergraph with the maximum number of colours such that each hyperedge contains vertices of all colours.
\end{definition}
Note that each colour in a polychromatic colouring of $\mathcal{H}_{U,\mathcal{S}}$ corresponds to a set cover of $U$ using subsets in $\mathcal{S}$. The fact that vertices must have different colours forces the set covers to be disjoint, and maximizing the number of colours maximizes the number of disjoint set covers, i.e., solves the DSCP.

\subsection{Deterministic Offline Algorithm \cite{bagaria2013optimally}} \label{sec_off}
The offline algorithm in \cite{bagaria2013optimally} solves the DSCP through the polychromatic colouring of $\mathcal{H}_{U,\mathcal{S}}$. Let us first refresh some notation from \cite{bagaria2013optimally}, and then present a slightly different exposition of the algorithm from \cite{bagaria2013optimally}. Let $V(e)$ denote the set of vertices in hyperedge $e$. An incomplete colouring of a hypergraph is one in which there exist colours that are not present in all hyperedges. These colours are said to be \emph{invalid}, since the subset collections that they correspond to do not form set covers. Let us colour the hypergraph using $\ell= F_{min}/\ln (n\ln n)$ colours, using the set $[\ell]$. Given an incomplete colouring of a hypergraph using colours in set $[\ell]$, we denote by random variable $L$ the number of invalid colours. We also define an indicator random variable $X_c$ for each $c\in [\ell]$, which is $1$ if colour $c$ is invalid and $0$ otherwise. The following relation is readily established, as
\begin{equation}
L = \sum_{c\in [\ell]}X_c. \label{eq1}
\end{equation}
We also define another indicator random variable $Y_{e,c}$ for each hyperedge-colour pair $(e\in E, c \in [\ell])$, which is $1$ if hyperedge $e$ does not contain any vertex coloured with colour $c$, and $0$ otherwise. The relation between $X_c$ and $Y_{e,c}$ is as follows:
\begin{equation}
X_c \leq \sum_{e\in E} Y_{e,c}, \text{     }\forall \text{ }c\in [\ell]. \label{eq2}
\end{equation}
And so, by \eqref{eq1} and \eqref{eq2},
\begin{equation}
L \leq \sum_{c\in [\ell]}\sum_{e\in E} Y_{e,c}. \label{rv}
\end{equation}
The offline algorithm of \cite{bagaria2013optimally} \texttt{PolyOff} operates in two phases. In Phase I, it colours all vertices uniformly randomly using $\ell = F_{min}/ \ln \ln n$ colours. At this point, note from \eqref{rv} that
\begin{align}
\mathbf{E}[L] \leq \sum_{c\in [\ell]}\sum_{e\in E} \mathbf{P}[Y_{e,c}=1] &\leq n\cdot\ell\cdot(1-\frac{1}{\ell})^{|V(e)|}, \nonumber\\
									&\leq n\ell e^{-|V(e)|/\ell} \leq n\ell e^{-\ln (n\ln n)}, \nonumber\\
									&=\ell/\ln n. \label{invalid}
\end{align}
This essentially achieves a randomized polychromatic colouring with at most $\ell/\ln n$ invalid colours. After this comes Phase II of the \texttt{PolyOff} algorithm - the recolouring phase - which works as follows to achieve a deterministic colouring. Order the vertices arbitrarily as $v_1, v_2,\ldots$ and \emph{recolour} them in this order. Let vertex $v_i$ be recoloured by colour $c_i$. Let $|V^{u}(e)|$ denote the number of vertices in hyperedge $e$ that have not been recoloured. Now for all hyperedges $e$, the probability that $e$ does not contain colour $c$ given that the vertices $v_{1}, \ldots, v_{i}$ have been recoloured with colours $c_{1}, \ldots, c_{i}$ is $0$ if there exists a vertex in $e$ which has already been recoloured with colour $c$. Otherwise, the probability is given by $\big(1-1/\ell\big)^{|V^{u}(e)|}$, since vertices $v_{i+1},\ldots,v_{|\mathcal{S}|}$ were each coloured uniformly randomly with one of the $\ell$ colours in Phase I. So,
\begin{equation} 
	\mathbf{P}[\,Y_{e,c}=1|c_{1}, c_{2}, \ldots, c_{i}]=
	\begin{cases}
	0 , \text{ if } \exists \text{ } q\leq i \text{ s.t. } v_{q}\in V(e) \\
	\qquad \qquad \qquad \text{ and } c_{q}=c\\
	\big(1-1/\ell\big)^{|V^{u}(e)|}  \text{ otherwise.}
	\end{cases} \label{probrecolour}
\end{equation}
After the vertices $v_1,v_2, \ldots, v_j$ are recoloured, we denote the conditional expectation of the number of invalid colours by $\mathbf{E}[L|c_1,c_2,\ldots, c_j]$. Note that 
\begin{equation*}
\mathbf{E}[\,L|c_{1},c_{2} ,\ldots, c_{i}]\leq \sum_{e\in E} \sum_{c\in [\ell]}\mathbf{P}\big[\,Y_{e,c}=1|c_{1}, c_{2}, \ldots, c_{i}\big].
\end{equation*}
Recolour $v_1$ uniformly randomly from $[\ell]$. Given that $v_1, \ldots, v_{i-1}$ have been recoloured, recolour $v_i$ such that
\vspace{-1mm}
\begin{equation}
\mathbf{E}[L|c_1, c_2,\ldots, c_{i-1}, c_i]\leq \mathbf{E}[L|c_1, c_2,\ldots, c_{i-1}] \label{recolour}
\vspace{-1mm}
\end{equation}
Note that such a recolouring exists, since due to the colouring of Phase I being uniformly random,
\vspace{-1mm}
\begin{equation*}
\mathbf{E}[L|c_1, c_2,\ldots, c_{i-1}]= (1/\ell)\sum_{c_i\in[\ell]} \mathbf{E}[L|c_1, c_2,\ldots, c_{i-1}, c_i],
\vspace{-1mm}
\end{equation*}
which is a convex combination. So there exists at least one colour $c_i$ for which \eqref{recolour} holds. Before recolouring (at the end of Phase I), \eqref{invalid} tells us that the number of invalid colours was less than $\ell/\ln n$, and we can ensure throughout the recolouring process that that number does not increase. Therefore, at the end of the algorithm, the number of colours that form a polychromatic colouring of $\mathcal{H}_{U,\mathcal{S}}$ is at least
\vspace{-1mm}
\begin{equation}
\ell-\frac{\ell}{\ln n} = \frac{F_{min}}{\ln n}\bigg(1 - \frac{\ln\ln n+1}{\ln(n\ln n)}\bigg). \label{actualapprox}
\vspace{-1mm}
\end{equation}
The vertices coloured with the invalid colours do not correspond to set covers. The \texttt{PolyOff} algorithm therefore obtains \linebreak \hbox{ $\frac{F_{min}}{\ln n}\Big(1 - o(1)\Big)$} disjoint set covers, where the $o(1)$ term is $<1/2$ and goes to zero as $n\rightarrow\infty$.

\vspace{-1mm}
\subsection{The Online Extension}  \label{onex}
Recall Remark \ref{mininf}, by which we will assume that all online algorithms know $F_{min}$ in advance. We will now extend ideas from the \texttt{PolyOff} algorithm presented in Section \ref{sec_off} to produce the \texttt{PolyOn} online algorithm.  

Before we do that however, we introduce a randomized online algorithm, $\texttt{RandColour}$, which colours each subset with one of $F_{min}/\ln n$ colours uniformly randomly on arrival. This effectively performs Phase I of the \texttt{PolyOff} algorithm on the incoming subsets and produces, in expectation, at least $F_{min}/\ln n\Big(1 - o(1)\Big)$ disjoint set covers, by \eqref{actualapprox}. 

We now make the following claim:
\begin{theorem}
There exists a deterministic online algorithm for the DSCP with competitive ratio $\ln n$. \label{existence}
\end{theorem}
We will prove Theorem \ref{existence} by constructing a deterministic online algorithm called the \texttt{PolyOn} algorithm. Note that the hypergraph $\mathcal{H}_{U, \mathcal{S}}$ is not available in advance to any online algorithm, and must be constructed as and when subsets arrive. To aid in this construction, we introduce a few concepts.
\vspace{-1mm}
\begin{definition}[Shrinking]
Shrinking a hyperedge $e$ of a hypergraph $G(V,E)$ to $k$ vertices (or size $k$) corresponds to removing elements of $V(e)$ such that $|V(e)|$ becomes equal to $k$. The removed elements of $V(e)$ can be arbitrary. Note that the vertex set $V$ is not disturbed; hyperedge $e$ is simply made to connect fewer vertices. \label{shrink}
\end{definition}
Let $i(e)$ represent the element $i\in U$ that was represented by hyperedge $e$ in $\mathcal{H}_{U,\mathcal{S}}$. Note that shrinking a hyperedge $e$ in $\mathcal{H}_{U,\mathcal{S}}$ to $k$ vertices corresponds to removing some $F_{i(e)}-k$ occurrences of the element $i(e)$ from the subsets that contain it. For example, in Figure \ref{Dual}, the hyperedge $1$ can be shrunk to size $1$ either by modifying the subset $\{1,2,4\}$ to $\{2,4\}$, or by modifying the subset $\{1,3\}$ to $\{3\}$.

We will now present the following Lemma.
\vspace{-1mm}
\begin{lemma}
Let the hyperedges in hypergraph $\mathcal{H}$ be shrunk to arbitrary sizes to produce another hypergraph $\mathcal{H}'$. Then any feasible polychromatic colouring of $\mathcal{H}'$ will be a feasible polychromatic colouring of $\mathcal{H}$. \label{hypshrink}
\end{lemma}
\begin{proof}
Let hyperedge edge $e\in\mathcal{H}$ be shrunk to $e'\in \mathcal{H}'$. Consider a feasible polychromatic colouring of $\mathcal{H}'$ with some $\mathbf{c}$ colours. By definition, all hyperedges $e'\in H'$ contain all $\mathbf{c}$ colours. By Definition \ref{shrink}, for every $e'$ in $\mathcal{H}'$, $V(e')\subseteq V(e)$ where $e$ is the corresponding hyperedge in $\mathcal{H}$. Therefore, each hyperedge $e\in \mathcal{H}$ contains all $\mathbf{c}$ colours. This is a feasible polychromatic colouring of $\mathcal{H}$.  
\end{proof}

Let a hypergraph $\mathcal{H}^{min}_{U,\mathcal{S}}$ represent the hypergraph constructed from $\mathcal{H}_{U,\mathcal{S}}$ by shrinking each of its hyperedges to size $F_{min}$. As a corollary to Lemma \ref{hypshrink}, we therefore have:
\begin{corollary}\label{cor1}
Any polychromatic colouring of $\mathcal{H}^{min}_{U,\mathcal{S}}$ with $k$ colours is a polychromatic colouring of $\mathcal{H}_{U,\mathcal{S}}$ with $k$ colours.
\end{corollary}

The \texttt{PolyOn} algorithm will aim to construct hypergraph $\mathcal{H}^{min}_{U,\mathcal{S}}$ and polychromatically colour it with $F_{min}/\ln n$ colours, online. Note that two things must be accomplished together: \textbf{(a)} Online shrinking of $\mathcal{H}_{U,\mathcal{S}}$ to produce ${H}^{min}_{U,\mathcal{S}}$ from $\mathbf{S}$, and \textbf{(b)} Online polychromatic colouring of ${H}^{min}_{U,\mathcal{S}}$.

The \texttt{PolyOn} algorithm will accomplish \textbf{(a)} by the following \texttt{OnlineShrink} method. Note that hypergraph $\mathcal{H}^{min}_{U,\mathbf{S}}$ can be constructed from a sequence of subsets $\mathbf{S}$ by constructing another sequence $\mathbf{S}^{min}$ as and when subsets in $\mathbf{S}$ arrive, such that $F_i(\mathbf{S}^{min})= F_{min}(\mathbf{S})$ $\forall$ $i\in U$. One can simply do this online by ignoring all occurrences of all elements $i$ after their $F_{min}(\mathbf{S})$th occurrence. To present this formally, let us denote by $\mathbf{S}_j$ the sub-sequence of $\mathbf{S}$ containing its first $j$ subsets and let $S_j$ denote the $j$th subset of $\mathbf{S}$. Similarly define $\mathbf{S}^{min}_j$ and subset $S^{min}_j$. For every subset arrival $S_j$, as long as $F_i(\mathbf{S}_j)\leq F_{min}(\mathbf{S})$ $\forall$ $i\in U$, we set $S^{min}_j=S_j$. If after some subset $S_j$ arrives, if $F_i(\mathbf{S}_j)> F_{min}(\mathbf{S})$ for some $i\in U$, we set $S^{min}_j= S_j\backslash \{i\}$. We can thereby ensure that after all subset arrivals, $F_i(\mathbf{S}^{min})= F_{min}(\mathbf{S})$ $\forall$ $i\in U$. We construct $\mathcal{H}^{min}_{U,\mathbf{S}}$ as $\mathcal{H}_{U,\mathbf{S}^{min}}$, in an online fashion. 

Since the shrinking of hypergraph $\mathcal{H}_{U,\mathcal{S}}$ to produce ${H}^{min}_{U,\mathcal{S}}$ can be accomplished online, from hereon, we will assume that the \texttt{OnlineShrink} process is carried out by the \texttt{PolyOn} algorithm for any arrival sequence of subsets. Now, we can assume that the subset sequence $\mathbf{S}^{min}$ arrives, and aim to achieve \textbf{(b)} simultaneously with \textbf{(a)}, i.e. to polychromatically colour $\mathcal{H}_{U,\mathbf{S}^{min}}$ (or $\mathcal{H}^{min}_{U,\mathbf{S}}$) online, as follows, by the \texttt{OnlineColour} method:

The \texttt{PolyOn} algorithm will first create the set of $F_{min}/ \ln (n\ln n)$ colours $[\ell]$. It will assume that all subsets (vertices) in $\mathbf{S}^{min}$ have already been coloured uniformly randomly with a colour from $[\ell]$ \emph{prior to arrival}. It will now recolour each incoming vertex with the prior knowledge that all hyperedges in $\mathcal{H}_{U,\mathbf{S}^{min}}$ have exactly $F_{min}$ vertices, i.e. $|V(e)|=F_{min}$ $\forall$ $e$. Recall Phase II of the \texttt{PolyOff} algorithm, in which we calculated the probability that hyperedge $e$ does not contain colour $c$ after vertices $v_1$ through $v_j$ had been recoloured, as $\mathbf{P}[\,Y_{e,c}=1|c_{1}, c_{2}, \ldots, c_{j}]$. The expression for $\mathbf{P}[\,Y_{e,c}=1|c_{1}, c_{2}, \ldots, c_{j}]$ was given by \eqref{probrecolour}. The \texttt{PolyOn} algorithm will attempt to recolour vertex $v_j$ similarly. but it is not provided $|V^u(e)|$ directly. It can, however, calculate $|V^u(e)|=|V(e)|-|V^d(e)|= F_{min}-|V^d(e)|$, where $|V^d(e)|$ is the number of vertices in hyperedge $e$ that have been recoloured. This is all information that is available to the online algorithm, and it can therefore calculate
\begin{equation*}
\mathbf{P}[\,Y_{e,c}=1|c_{1}, c_{2}, \ldots, c_{j}]=
	\begin{cases}
	0 , \text{ if } \exists \text{ } q\leq i \text{ s.t. } v_{q}\in V(e) \\
	\qquad \qquad \qquad \text{ and } c_{q}=c\\
	\big(1-1/\ell\big)^{F_{min}-|V^d(e)|}  \text{ else.}
	\end{cases} \label{probon}
\end{equation*}
for each hyperedge-colour pair $(e,c)$ in an online fashion! The \texttt{PolyOn} algorithm can then use \eqref{probon} to calculate the conditional expectation of the number of invalid colours $\mathbf{E}[\,L|c_{1},c_{2} ,\ldots, c_{i}]\leq \sum_{e\in E} \sum_{c\in [\ell]}\mathbf{P}\big[\,Y_{e,c}=1|c_{1}, c_{2}, \ldots, c_{i}\big]$. It will then recolour vertex $v_{i+1}$ such that $\mathbf{E}[L|c_1, c_2,\ldots, c_{i-1}, c_{i+1}]\leq \mathbf{E}[L|c_1, c_2,\ldots, c_{i}]$, $\forall$ $i$. Note that after all vertex arrivals, this will result in a polychromatic colouring of the hypergraph $\mathcal{H}_{U,\mathbf{S}^{min}}$ with $F_{min}/\ln n$ colours, by an equation similar to \eqref{actualapprox}. By Corollary \ref{cor1}, this is also a polychromatic colouring of the hypergraph $\mathcal{H}_{U,\mathbf{S}}$ with $F_{min}/\ln n$ colours, and so the \texttt{PolyOn} algorithm has produced $F_{min}/\ln n$ disjoint set covers of $U$ from $\mathbf{S}$ in an online manner, by executing both the \texttt{OnlineShrink} and \texttt{OnlineColour} operations together.

We know that the \texttt{PolyOn} algorithm (which we represent now by $\mathcal{P}$) returns at least $F_{min}/ \ln n$ set covers for all input sequences $\mathbf{S}$. Using \eqref{upperbound}, we can see that
\begin{equation*}
\mu(\mathcal{P}) = \max_{\mathbf{S}}\frac{T(\mathcal{M}^*_{off}, \mathbf{S})}{T(\mathcal{M}_{\mathcal{P}}, \mathbf{S})} = \frac{T(\mathcal{M}^*_{off}, \mathbf{S})}{\min_{\mathcal{S}}T(\mathcal{M}_{\mathcal{P}}, \mathbf{S})} \leq \ln n.
\end{equation*}
This proves Theorem \ref{existence}.

\subsection{Advantages of the \texttt{PolyOn} Online Algorithm}
Firstly, notice that for all input sequences in which $F_{min}\geq \ln n$, the \texttt{PolyOn} algorithm performs better than the trivial \texttt{GreedyCover} algorithm $\mathcal{V}$. 

The second advantage is that the \texttt{PolyOn} algorithm is polynomial time optimal.
It was shown in \cite{feige2002approximating} that no polynomial time algorithm (offline or online) can produce more than $T(\mathcal{M}^*_{off}, \mathbf{S})/\ln n$ disjoint set covers $\forall$ $\mathbf{S}$ unless $NP\subseteq DTIME(n^{\BigO{\log\log n}})$. Therefore, no online algorithm that runs in polynomial time can have a better competitive ratio than the \texttt{PolyOn} algorithm. The \texttt{PolyOn} algorithm also matches the performance of the best polynomial time offline algorithm. 

The following Lemma confirms the third advantage.
\begin{lemma}
There exists an infinite family of subset sequences $\mathbf{S}$ for which the competitive ratio of the \texttt{PolyOn} algorithm is $1$, i.e. $\mu_{\mathbf{S}}(\mathcal{P})=1$.
\end{lemma}
\begin{proof}
We will show an infinite family of subset sequences $\mathbf{S}$ for which $T(\mathcal{M}^*_{off}, \mathbf{S}) \leq F_{min}/\ln n$. Note that for these sequences, $\mu_{\mathbf{S}}(\mathcal{P}) = 1$. To show this, consider the \emph{minimum set cover} problem (\texttt{MinSetCover}), in which given $U$ and $\mathcal{S}$, we are required to find the set cover $C\subseteq \mathcal{S}$ such that $|C|$ is minimized. The integer programming (IP) formulation of \texttt{MinSetCover} is known to have an integrality gap $\ln n$ \cite{feige1998threshold}, or in other words, for any problem instance $\mathcal{I}$, the ratio of the IP's optimal solution ($OPT(IP_\mathcal{I})$) to the LP relaxation's optimal solution ($OPT(LP_\mathcal{I})$), in which each subset is given a coefficient between $0$ and $1$, is at most $\ln n$. Consider one such $\ln n$ integrality gap instance $\mathcal{I}_0$ in which the frequency of all elements is equal to $F_{min}$, and $\frac{OPT(IP_{\mathcal{I}_0})}{OPT(LP_{\mathcal{I}_0})}=\ln n$. For an example of such an instance $\mathcal{I}_0$, we refer the reader to Example 13.4 of \cite{vazirani2001approximation}. Note that $OPT(LP_{\mathcal{I}_0})=|\mathcal{S}|/ F_{min}$, which is accomplished by setting the coefficient of each subset in $\mathcal{S}$ to $1/F_{min}$. The optimal IP solution $OPT(IP_{\mathcal{I}_0})$  will therefore be at least $\frac{|\mathcal{S}|}{F_{min}}\ln n$, since this is a $\ln n$ integrality gap instance. Each set cover in this instance therefore consists of at least $N = \frac{|\mathcal{S}|}{F_{min}}\ln n$ subsets, and so the number of disjoint set covers can be at most $|\mathcal{S}|/N = F_{min}/\ln n$. We have therefore shown an infinite family (for different $n$) for which $T(\mathcal{M}^*_{off}, \mathbf{S}) \leq F_{min}/\ln n$ and $\mu_{\mathbf{S}}(\mathcal{P}) = 1$.
\end{proof}

\section{The Lower Bound}
In this section, for ease of exposition, we will first show a lower bound of $\Omega((\ln n)^{1/3})$ on the competitive ratio of any online algorithm, which we will then improve to $\Omega((\ln n)^{1/2})$.
We first present a brief overview of our method.

\begin{overview}
We will generate a subset sequence for an online algorithm consisting of 3 sub-sequences - let us call them $X$, $Y$ and $Z$ - in that order. We will first provide the sequence $X$, on which an online algorithm $\mathcal{A}$ will make an allocation $\mathcal{M}_\mathcal{A}(X)$. Depending on $\mathcal{M}_\mathcal{A}(X)$, we will then provide an adversarial sequence $Y(\mathcal{M}_\mathcal{A}(X))$ that upper bounds the number of disjoint set covers that can be formed by $\mathcal{A}$. Lastly, we will provide sequence $Z(\mathcal{M}_\mathcal{A}(X))$ such that the optimal offline algorithm can construct a larger number of disjoint set covers by making reallocations, which is not allowed for $\mathcal{A}$. We thus obtain a lower bound on the competitive ratio of all online algorithms. \label{ex1}
\end{overview}

In this section, we define the universe as $U=\{0,1\}^q$. There are therefore a total of $n=2^q$ elements, in which the $i$th element $u_i$ is represented by the $q$-bit binary representation of $i$, and $0\leq i \leq 2^q-1$. We also define $u^k_i$ as the $k$th bit of $u_i$, for $0\leq k \leq q-1$. We now form a sequence of subsets $\mathbf{S}_{com} = [S_j: S_j=\{u_i:u^j_i = 1\}, 1\leq j\leq q]$, i.e. subset $j$ of $\mathbf{S}_{com}$ contains all elements $u_i$ that have $1$ in their $j$th position. Note that $|\mathbf{S}_{com}| = q$ and that $|S_j|= 2^{q-1}$ $\forall$ $S_j\in \mathbf{S}_{com}$. We also see by definition that the frequency of element $u_i$ in $\mathbf{S}_{com}$, $F_i(\mathbf{S}_{com}) = \#\{k: u^k_i=1\} = q - \#\{k: u^k_i=0\}$. Also, $T(\mathcal{M}^*_{off},\mathbf{S}_{com})=0$, since $F_{min}(\mathbf{S}_{com})=0$. This is because the all-zero element does not appear in any of the subsets in $\mathbf{S}_{com}$. Therefore, $\mathbf{S}_{com}$ does not contain any set cover.

We will use $\mathbf{S}_{com}$ as the start of the subset sequence, i.e. sequence $X$ in Overview \ref{ex1}, in the proofs of Theorems \ref{root3} and \ref{root2}. To that end, note the following. Let any online algorithm $\mathcal{A}$ make an allocation $\mathcal{M}_{\mathcal{A}}(\mathbf{S}_{com})$ on $\mathbf{S}_{com}$, and partition its subsets into $P_1, P_2,\ldots$. After all subsets in $\mathbf{S}_{com}$ have arrived and been allocated, denote the elements $u_i\in U$ that are missing from partition $P_j$ by the set $E_j(\mathcal{M})$. We can assume the following without loss of generality, for all allocations $\mathcal{M}$
\begin{align}
E_1(\mathcal{M})=\{u_i:u^k_i = 0, k = \{0,1,\ldots r(1)\} \nonumber \\
\qquad \qquad \qquad \qquad \qquad \qquad		\text{ for some } r(1)\geq 0\}. \label{part1} 
\end{align}
In words, partition $P_1$ does not contain all those elements $u_i$ that have zeroes in their first $r(1)$ places, and hence contains the first $r(1)+1$ subsets in $\mathbf{S}_{com}$. This can always be accomplished by reordering the $q$ bits. Also, by the definition of subsets in $\mathbf{S}_{com}$, the number of subsets in $P_1$ is $r(1)+1$. Similar to \eqref{part1}, we see that 
\begin{align}
E_2(\mathcal{M})=\{u_i:u^k_i = 0, k = \{r(1)+1,\ldots r(2)\} \nonumber\\
\qquad \qquad \qquad \qquad \qquad \qquad                    \text{ for some } r(2)> r(1)\}, \label{part2}\\
E_j(\mathcal{M})=\{u_i:u^k_i = 0, k = \{r(j)+1,\ldots r(j+1)\} \nonumber\\
\qquad \qquad \qquad \qquad \text{ for some } r(j+1)> r(j)\} \text{ } \forall \text{ }\mathcal{M}. \label{part3}
\end{align}
Let us denote the number of subsets in partition $P_j$ by $d_j= r(j)- r(j-1)$ $\forall$ $j\geq 1$ (we define $r(0)=-1$ for consistency). Without loss of generality, we can assume that $d_{j+1}\geq d_j$ for all allocations $\mathcal{M}$, by reordering partitions if necessary.

After an allocation $\mathcal{M}_{\mathcal{A}}(\mathbf{S}_{com})$, let the maximum number of subsets placed in any partition be $L$, where $0\leq L\leq q$. Define the sets $D_{\ell}(\mathcal{M})=\{j:d_j=\ell\}$ for each $1\leq\ell\leq L$. In words, $D_{\ell}$ represents the indices of those partitions that contain exactly $\ell$ subsets from $\mathbf{S}_{com}$ after an online allocation $\mathcal{M}$. To simplify notation, we will use $D_\ell$ to represent $D_\ell(\mathcal{M})$ when there is no ambiguity about the allocation $\mathcal{M}$.

After the allocation $\mathcal{M}_{\mathcal{A}}(\mathbf{S}_{com})$ is done, we also define a \emph{bottleneck element} for each partition $P_j$ as $u_{b(j)}$ such that $u_{b(j)}\in E_j(\mathcal{M})$ and $\#\{k:u^k_{b(j)}=0\}$ is minimum, i.e. $u_{b(j)}$ is the element with the fewest zeroes that does not appear in partition $j$. For example, take an allocation in which $d_1=1$ and $d_2=2$, i.e. partition $P_1$ contains one subset $S_1=\{u_i:u^1_i = 1\}$ and partition $P_2$ contains $2$ subsets $S_2=\{u_i:u^2_i = 1\}$ and $S_3=\{u_i:u^3_i = 1\}$. In that case, the bottleneck elements for partitions $P_1$ and $P_2$ are $u_{b(1)}=01111\ldots q-1 \text{ times}$ and $u_{b(2)}=10011\ldots q-3 \text{ times}$, respectively. All elements of the universe that are not bottleneck elements are called \emph{non-bottleneck elements}. Note that the frequency in $\mathbf{S}_{com}$ of a bottleneck element of partition $j$ is $F_{b(j)}(\mathbf{S}_{com})=q - d_j$. We will combine these bottleneck elements appropriately to form the next sequence (corresponding to $Y$ in Overview \ref{ex1}) that upper bounds the number of disjoint set covers obtainable by any online algorithm.

Note that all the quantities have been defined with respect to allocations of $\mathbf{S}_{com}$.
We are now ready to prove that:
\begin{theorem}
The competitive ratio of any online DSCP algorithm is lower bounded by $\Omega((\ln n)^{1/3})$, i.e. $\mu(\mathcal{A})= \Omega((\ln n)^{1/3})$ $\forall$ $\mathcal{A}$. \label{root3}
\end{theorem}
\begin{proof}
We will provide a subset sequence $\mathbf{S}_1$ for which $T(\mathcal{M}^*_{off},\mathbf{S}_1)=\Omega(F_{min}(\mathbf{S}_1))$, and $T(\mathcal{M}_{\mathcal{A}},\mathbf{S}_1)= \BigO{\frac{F_{min}(\mathbf{S}_1)}{(\ln n)^{1/3}}}$ $\forall$ online algorithms $\mathcal{A}$. The first $q$ elements of $\mathbf{S}_1$ will be identical to $\mathbf{S}_{com}$. We will then create the remaining elements of $\mathbf{S}_1$ adversarially depending on the allocation of $\mathbf{S}_{com}$ in order to limit the number of disjoint set covers that can be created by any online algorithm.

Consider an online algorithm $\mathcal{A}$ that makes an allocation $\mathcal{M}_{\mathcal{A}}(\mathbf{S}_{com})$ on the first $q$ subsets. We will assume \eqref{part1}, \eqref{part2} and \eqref{part3}, which hold without loss of generality. We now construct the next sequence of subsets $\mathbf{S}_a$ adversarially. Create one copy of the subset $S^a_{1}= \{x:x=u_{b(j)}, j\in D_{1}\}$, which contains bottleneck elements of all partitions that contain exactly $1$ subset from $\mathbf{S}_{com}$. Note that $F_{b(j)}(\mathbf{S}_{com}\wedge S^a_1) = q$ $\forall$ $j\in D_1$. We will ensure that these bottleneck elements $u_{b(j)}$, $j\in D_{1}$ never arrive later in sequence $\mathbf{S}_1$. Now look at all the partitions $P_j$ where $j\in D_1$. They all require their respective bottleneck elements in order to form set covers, and yet there is only one available subset $S^a_1$ that contains these bottleneck elements. Therefore, at most $1$ partition among $P_j$, $j\in D_1$ can form a set cover. 
Extending this further, we now create the adversarial sequence $\mathbf{S}_a$ containing $\ell$ copies of the subsets $S^a_{\ell}= \{x:x=u_{b_j}, j\in D_{\ell}\}$ for all $1\leq\ell\leq L$, in some arbitrary order. Note that through this, $F_{b(j)}(\mathbf{S}_{com}\wedge \mathbf{S}_a) = q$, $\forall$ $j$. Similar to the argument for partitions $P_j$ where $j\in D_1$, we can argue that a maximum of $\ell$ partitions among $P_j$ for $j\in D_\ell$ can form set covers. Also define a sequence $\mathbf{S}_{inf}$ which contains an arbitrarily large number of singleton subsets of all non-bottleneck elements in any arbitrary order. Let $\mathbf{S}_1=\mathbf{S}_{com}\wedge \mathbf{S}_a\wedge \mathbf{S}_{inf}$.

Note that $F_{min}(\mathbf{S}_1)=F_{b(j)}(\mathbf{S}_1)$ for any $j$, and $F_{b(j)}(\mathbf{S}_1)= q$ $\forall$ $j$, since the non-bottleneck elements appear an infinite number of times. Now that we have constructed $\mathbf{S}_1$ depending on $\mathcal{M}_{\mathcal{A}}(\mathbf{S}_{com})$, we present the following Lemma.
\begin{lemma}
The number of disjoint set covers that can be formed by any online algorithm from $\mathbf{S}_1$ is $\BigO{q^{2/3}}$, i.e. $T(\mathcal{M}_{\mathcal{A}}, \mathbf{S}_1)= \BigO{q^{2/3}}$ $\forall$ online algorithms $\mathcal{A}$. \label{lemma1}
\end{lemma}
\begin{proof}
Let $A_\ell$ denote the number of disjoint set covers that can be formed by partitions that have $\ell$ subsets from $\mathbf{S}_{com}$. Thus, $A_\ell=\#\{P_j:j\in D_\ell, P_j \text{ forms a set cover.}\}$. By definition, $A_\ell\leq |D_\ell|$. Recall that we have ensured by providing $\mathbf{S}_a$ that a maximum of $\ell$ partitions among $P_j$ for $j\in D_\ell$ can form set covers for any allocation $\mathcal{M}_\mathcal{A}(\mathbf{S}_{com})$. In other words, $A_\ell\leq \min(\ell, |D_\ell|)$ partitions can form set covers for each $D_\ell$. The optimization problem over all possible online allocations $\mathcal{M}_\mathcal{A}(\mathbf{S}_{com})$ is therefore the following:
\vspace{-2mm}
\begin{eqnarray}
&&\text{Maximize : } \qquad \sum_{\ell=1}^L A_\ell \label{opt1}\\
&&\text{Subject to } \qquad \sum_{\ell=1}^{L}\ell\cdot|D_\ell| = q, \label{const1a}\\
&&\qquad \qquad A_\ell \leq \min(\ell, |D_\ell|) \text{ } \forall \text{ } 1\leq \ell\leq L \label{const2a}
\vspace{-3mm}
\end{eqnarray}
where \eqref{opt1} corresponds to maximizing the number of disjoint set covers, \eqref{const1a} is the constraint on the number of subsets available in $\mathbf{S}_{com}$, and \eqref{const2a} is the constraint imposed by the bottleneck elements in $\mathbf{S}_a$. The optimization is over $|D_\ell|$ (each $|D_\ell|=|D_\ell(\mathcal{M})|$ corresponds to some allocation $\mathcal{M}_\mathcal{A}(\mathbf{S}_{com})$.

The optimal solution to \eqref{opt1} occurs when $|D_\ell|=\ell$, $\forall$ $\ell$. That would imply that the optimal online allocation corresponds to creating $\ell$ partitions $P_j$ each containing $\ell$ subsets from $\mathbf{S}_{com}$, i.e. $1$ partition with $1$ subset, $2$ partitions with $2$ subsets and so on. Even after the sequence $\mathbf{S}_a$ is provided, each partition created by $\mathcal{M}_\mathcal{A}(\mathbf{S}_{com})$ can form a set cover. This is because $\mathbf{S}_a$ can provide each partition with its bottleneck element, and $\mathbf{S}_{inf}$ provides all non-bottleneck elements. Therefore, the constraint \eqref{const1a} evaluates to 
\vspace{-2mm}
\begin{equation}
\sum_{\ell=1}^{L}\ell\cdot|D_\ell| = \sum_{\ell=1}^{L}\ell^2=q, 
\vspace{-2mm}
\end{equation}
which implies that $\BigO{L^3}=q$, so $L= \BigO{q^{1/3}}$. The number of disjoint set covers is therefore
\vspace{-2mm}
\begin{equation}
\sum_{\ell=1}^L \: \min(\ell, |D_\ell|) = \sum_{\ell=1}^L \: \ell = \BigO{L^2} = \BigO{q^{2/3}}.
\vspace{-2mm}
\end{equation}
Therefore, $T(\mathcal{M}_{\mathcal{A}}, \mathbf{S}_1)= \BigO{q^{2/3}}$ $\forall$ online algorithms $\mathcal{A}$.
\end{proof}

\begin{lemma}
The number of disjoint set covers that can be formed from $\mathbf{S}_1$ by the optimal offline algorithm is lower bounded by $q/2$, i.e. $T(\mathcal{M}^*_{off}, \mathbf{S}_1)\geq q/2$. \label{lemma2}
\end{lemma}
\begin{proof}
Since $\mathbf{S}_1$ was created adversarially based on the allocation made by the online algorithm, we will find the optimal offline solution by reallocating subsets from the online solution. Consider an allocation $\mathcal{M}_\mathcal{A}(\mathbf{S}_1)$. The following proposition will clarify how the reallocation must be done.
\begin{proposition}
The union of any two subsets $S_x,S_y\in \mathbf{S}_{com}$ taken such that $S_x\in P_s$ and $S_y\in P_t$ and $s\neq t$, will contain all bottleneck elements $u_{b(j)}$ $\forall$ $j$. \label{prop:disjoint}
\end{proposition}
\begin{proof}
Note that by definition, $S_x\cup S_y$ contains all bottleneck elements $u_{b(j)}$ $\forall$ $j\neq \{s,t\}$. In order to show that it also contains $u_{b(s)}$ and $u_{b(t)}$, we will show that $u_{b(s)}\in S_y$ and $u_{b(t)}\in S_x$. For some $g\neq h$, by definition, $S_x=\{u_i:u^g_i=1\}$, and $S_y=\{u_i:u^h_i=1\}$. Also, $u_{b(s)}\in E_s(\mathcal{M}$), and therefore by \eqref{part3}, $u_{b(s)}$ is such that $u^k_{b(s)} = 0$ $\forall$ $k = \{r(s)+1,\ldots r(s+1)\}$ and $u^k_{b(s)}=1$ otherwise. Similarly, $u_{b(t)}$ is such that $u^k_{b(t)} = 0$ $\forall$ $k = \{r(t)+1,\ldots r(t+1)\}$ and $u^k_{b(t)}=1$ otherwise. Now note that $g\not\in \{r(t)+1,\ldots r(t+1)\}$ and $h\not\in \{r(s)+1,\ldots r(s+1)\}$, since partitions $P_s$ and $P_t$ are distinct. Therefore, $u_{b(s)}\in S_y$ and $u_{b(t)}\in S_x$.
\end{proof}
From Proposition \ref{prop:disjoint}, it is clear that if an allocation $\mathcal{M}_\mathcal{A}(\mathbf{S}_{com})$ is such that subsets can be chosen pairwise from different partitions, the offline algorithm can produce at least $|\mathbf{S}_{com}|/2 = q/2$ disjoint set covers. This is because all bottleneck elements corresponding to the allocation $\mathcal{M}_\mathcal{A}(\mathbf{S}_{com})$ are covered by the union of two subsets from different partitions, and $\mathbf{S}_{inf}$ provides an infinite supply of non-bottleneck elements. However, for allocations $\mathcal{M}_\mathcal{A}(\mathbf{S}_{com})$ such that $d_j>|\mathbf{S}_{com}|/2$ for some partition $j$, such a pairwise choice of subsets from different partitions is impossible. A proof for that case is provided in the Appendix. 
\end{proof}
From Lemmas \ref{lemma1} and \ref{lemma2}, $\mu_{\mathcal{S}_1}(\mathcal{A})= \Omega(q^{1/3})$ $\forall$ online algorithms $\mathcal{A}$. Since $q= \log_2 n$, Theorem \ref{root3} is proved.
\end{proof}

We presented Theorem \ref{root3} to show that the grouping of bottleneck elements to form the subset sequence $\mathbf{S}_a$ allows us to lower bound $T(\mathcal{M}_{\mathcal{A}}, \mathbf{S}_1)$. We will build on these ideas in the presentation of Theorem \ref{root2}.
\begin{theorem}
The competitive ratio of any online DSCP algorithm is lower bounded by $\Omega((\ln n)^{1/2})$, i.e. $\mu(\mathcal{A})= \Omega((\ln n)^{1/2})$ $\forall$ $\mathcal{A}$. \label{root2}
\end{theorem}
\begin{proof}
We will use a technique similar to the proof of Theorem \ref{root3}. Instead of creating the subset sequence $\mathbf{S}_1$, however, we will now create the subset sequence $\mathbf{S}_2$, using a different adversarial sequence $\mathbf{S}_b$. The first $q$ subsets in $\mathbf{S}_2$ are still the sequence $\mathbf{S}_{com}$. Depending on the allocation $\mathcal{M}_\mathcal{A}(\mathbf{S}_{com})$, we will use the subsets in sequence $\mathbf{S}_a$ as defined in the proof of Theorem \ref{root3} to create the adversarial subsets in $\mathbf{S}_b$. Let $S^b_1=S^a_1\cup S^a_2\cup \ldots S^a_L$, $S^b_2=S^a_2\cup S^a_3\cup \ldots S^a_L$, and so on, with $S^b_\ell = \cup^{L}_{r = \ell}S^a_\ell$. Let $\mathbf{S}_b = [S_j: S_j=S^b_j, 1\leq j \leq L]$. Again, let $\mathbf{S}_{inf}$ be the infinite sequence of non-bottleneck elements for the allocation $\mathcal{M}_\mathcal{A}(\mathbf{S}_{com})$. Let $\mathbf{S}_2=\mathbf{S}_{com}\wedge \mathbf{S}_b\wedge \mathbf{S}_{inf}$.

Note: $F_{min}(\mathbf{S}_2) = F_{b(j)}(\mathbf{S}_2)$ for any $j$, $F_{b(j)}(\mathbf{S}_2)= q$ $\forall$ $j$.
\begin{lemma}
The number of disjoint set covers that can be formed by any online algorithm from $\mathbf{S}_2$ is $\BigO{q^{1/2}}$, i.e. $T(\mathcal{M}_{\mathcal{A}}, \mathbf{S}_2)= \BigO{q^{1/2}}$ $\forall$ online algorithms $\mathcal{A}$. \label{lemma1a}
\end{lemma}
\begin{proof}
Let $A_\ell=\#\{P_j:j\in D_\ell, P_j \text{ forms a set cover.}\}$, defined as before. Note that the total number of disjoint set covers that can be formed is upper-bounded by $L$, since there are only $L$ subsets in $\mathbf{S}_b$ that contain bottleneck elements, and each partition requires its bottleneck element in order to form a set cover. Note that only $1$ partition $P_j$, $j\in D_1$ can be made into a set cover, i.e. $A_1\leq 1$, since only one subset $S^b_1\in\mathbf{S}_b$ contains the bottleneck elements $u_{b(j)}$ $\forall$ $j\in D_1$. After that set cover is formed, only $1$ partition $P_j$, $j\in D_2$ can be made into a set cover by using the subset $S^b_2\in\mathbf{S}_b$, which is the only remaining subset in $\mathbf{S}_b$ that contains the bottleneck elements $u_{b(j)}$ $\forall$ $j\in D_2$. Alternatively, an online algorithm could have allocated $S^b_1$ and $S^b_2$ to $2$ partitions $P_j$, $j\in D_2$ and made them set covers. Mathematically, $A_2\leq 2 - A_1$. This logic can be extended for all partitions $P_j$, $j\in D_\ell$ $\forall$ $3\leq \ell\leq L$, to give rise to the constraint $A_\ell\leq \min(|D_\ell|, \ell - \sum^{\ell-1}_{x=1}A_x)$. The optimization problem over all possible online allocations $\mathcal{M}_\mathcal{A}(\mathbf{S}_{com})$ is therefore the following:
\begin{eqnarray}
&&\text{Maximize : } \qquad \sum^L_{\ell=1}A_\ell \label{opt2}\\
&&\text{Subject to } \qquad \sum_{\ell=1}^{L}\ell\cdot|D_\ell| = q, \label{const1} \\
&&  A_\ell\leq \min(|D_\ell|, \ell - \sum^{\ell-1}_{x=1}A_x) \text{ }\forall \text{ } 1\leq\ell\leq L. \label{const2}
\end{eqnarray}
The objective function \eqref{opt2} represents the maximization of the number of disjoint set covers. Constraint \eqref{const1} is because of the number of subsets available in $\mathbf{S}_{com}$, and constraint \eqref{const2} arises out of the structure of subsets in $\mathbf{S}_b$, in the fashion explained above. The optimal solution to this problem occurs for $|D_\ell|=1$ $\forall$ $1\leq \ell\leq L$, for which $A_\ell =1$ $\forall$ $1\leq \ell\leq L$. It is possible to intuitively see the reason for this, since an allocation with constraint \eqref{const1} will try to maximize the number of partitions that contain only $1$ subset from $\mathbf{S}_{com}$, and then two subsets from $\mathbf{S}_{com}$, and so on, after which constraint \eqref{const2} will ensure that a maximum of one partition containing $\ell$ subsets from $\mathbf{S}_{com}$ can form a set cover. For this solution, it is clear from constraint \eqref{const1} that $L = \BigO{q^{1/2}}$, and is the number of disjoint set covers.
\end{proof}

Note that $T(\mathcal{M}^*_{off},\mathbf{S}_2)= q/2$. The proof is identical to that of Lemma \ref{lemma2}, since $\mathbf{S}_2$ also contains $\mathbf{S}_{com}$ and $\mathbf{S}_{inf}$. Like with Lemma \ref{lemma2}, there is a slight technicality, which is dealt with in the Appendix.
So from Lemma \ref{lemma1a}, $\mu_{\mathcal{S}_2}(\mathcal{A})= \Omega(q^{1/2})$ $\forall$ online algorithms $\mathcal{A}$, where $q=\log_2 n$.
\end{proof}

\section{Simulations}
For lack of space, we present only one simulation result for the \texttt{PolyOn} algorithm $\mathcal{P}$. We considered the online resource allocation problem, in which each server acquires the files uniformly randomly, each with probability $p$, from the universe of $n$ files to form the subset sequence $\mathbf{S}$. Note that $E[F_{min}(\mathbf{S})]= |\mathbf{S}|p=k$ (say). We then appended subsets to ensure that $F_{min}=k$. Simulations were carried out for three different values of $F_{min}$. The plot of the number of set covers returned by the \texttt{PolyOn} algorithm $T(\mathcal{M}_\mathcal{P},\mathbf{S})$ versus $n$ is provided in Figure \ref{sim}. We can see that we get approximately $F_{min}/ \ln n$ set covers for all the three scenarios, thus validating our theoretical analysis in Section \ref{onex}.
\begin{figure}[h!]
    \begin{center}
    \psfrag{0}{\footnotesize{$0$}}
    \psfrag{50}{\footnotesize{$50$}}
    \psfrag{100}{\footnotesize{$100$}}
    \psfrag{150}{\footnotesize{$150$}}
    \psfrag{200}{\footnotesize{$200$}}
    \psfrag{250}{\footnotesize{$250$}}
    \psfrag{300}{\footnotesize{$300$}}
    \psfrag{350}{\footnotesize{$350$}}
    \psfrag{400}{\footnotesize{$400$}}
    \psfrag{450}{\footnotesize{$450$}}
    \psfrag{10}{\footnotesize{$10$}}
    \psfrag{30}{\footnotesize{$30$}}
    \psfrag{50}{\footnotesize{$50$}}
    \psfrag{70}{\footnotesize{$70$}}
    \psfrag{90}{\footnotesize{$90$}}
    \psfrag{110}{\footnotesize{$110$}}
    \psfrag{130}{\footnotesize{$130$}}
    \psfrag{150}{\footnotesize{$150$}}
    \psfrag{170}{\footnotesize{$170$}}
    \psfrag{190}{\footnotesize{$190$}}
    \psfrag{aaaaaaaaaaaa}[][][.65]{$F_{min}=1000$}
    \psfrag{ccccccccccccc}[][][.65]{$F_{min}=750$}
    \psfrag{eeeeeeeeeeee}[][][.65]{$F_{min}=500$}
    \psfrag{xlabel}[][][0.9]{Number of files $n$}
    \psfrag{ylabel}[][][0.9]{$T(\mathcal{M}_\mathcal{P},\mathbf{S})$}
    \includegraphics[scale=.58]{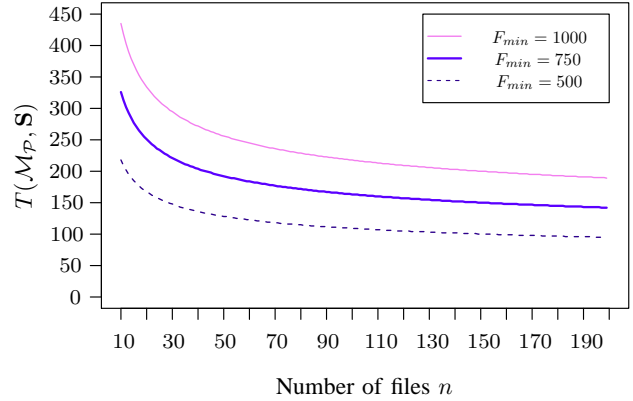}
    \caption{Plot of the number of set covers obtained by the \texttt{PolyOn} algorithm $T(\mathcal{M}_\mathcal{P},\mathbf{S})$ versus the number of files $n$ for $3$ values of $F_{min}$.} \label{sim}
    \end{center}
\end{figure}
\section{Concluding Remarks}
We presented and analysed online algorithms for the DSCP in terms of their worst case competitive ratios. We found a lower bound on the competitive ratio for all online algorithms, and presented an online algorithm that performed comparable to that lower bound for reasonable $n$. We conjecture that the tight lower bound on the competitive ratio is $\ln n$, but it will most likely require an entirely different approach to show. Analysis of the average case competitive ratio of online algorithms for the DSCP is still an open problem.

The results of this paper can be extended to produce online algorithms with competitive ratio $\ln n$ for all problems that involve finding disjoint bases in a polymatroid \cite{cualinescu2009disjoint}. Some examples of such problems are the domatic number problem \cite{feige2002approximating}, which has applications in the connectivity of WSNs \cite{moscibroda2005maximizing}, and in packing element-disjoint Steiner trees \cite{cheriyan2007packing}. The lower bound on the competitive ratio also carries over to the general polymatroid problem of \cite{cualinescu2009disjoint}.

\bibliographystyle{IEEEtran}
\bibliography{research}

\begin{appendix} \label{app}
\textbf{Offline Solution for the special case:} When $\mathcal{M}_\mathcal{A}(\mathbf{S}_{com})$ was such that $d_j\geq q/2$ for some $j=k$ (say), the offline solution could no longer be found by pairing subsets from different partitions. Note that only one such partition $P_k$ can exist. For this case alone, after the allocation $\mathcal{M}_\mathcal{A}(\mathbf{S}_{com})$, we generate the adversarial subset sequence (either $\mathbf{S}_a$ or $\mathbf{S}_b$), differently. We consider partition $P_k$ to consist of $2$ partitions $P_{k_1}$ and $P_{k_2}$, each of size less than $q/2$, and consider each to have its own bottleneck element. We then construct the adversarial sequence with this assumption of an additional bottleneck element. Now, all online algorithms are subject to all the constraints of \eqref{opt1} (and \eqref{opt2}), except that they can, in addition, make $P_k$ a set cover. Therefore, the statements of Lemmas \ref{lemma1} (and \ref{lemma1a}) still hold. The offline algorithm, however, will produce $q/2$ disjoint set covers, where $2$ subsets can now be chosen pairwise from $P_{k_1}$ and $P_{k_2}$. For this case too, Theorems \ref{root3} and \ref{root2} hold.
\end{appendix}
\end{document}